%% file: paper.tex
\documentclass[a4paper,USenglish,final]{lipics-v2018}
\usepackage{microtype}
\usepackage{xspace}
\usepackage{algorithm}
\usepackage[noend]{algpseudocode}
\usepackage{dsfont}

\definecolor{blueLink}{rgb}{0,0.2,0.8}
\hypersetup{colorlinks,linkcolor=blueLink,urlcolor=blueLink,citecolor=blueLink}

\newcommand{\I}{\ensuremath{\mathcal{I}}\xspace}
\newcommand{\R}{\ensuremath{V}\xspace}
\newcommand{\E}{\ensuremath{E}\xspace}
\newcommand{\ALGgeneric}{\textsc{Alg}\xspace}
\newcommand{\ALG}{\textsc{GD}\xspace}
\newcommand{\ALGlong}{\textsc{Greedy Dual}\xspace}
\newcommand{\OPT}{\textsc{Opt}\xspace}
\newcommand{\TB}{\textsc{TB}\xspace}
\newcommand{\dist}{\textsf{dist}}
\newcommand{\pos}{\textsf{pos}}
\newcommand{\atime}{\textsf{atime}}

\newcommand{\Free}{\textsf{free}}
\newcommand{\Moat}{\ensuremath{\mathcal{A}}\xspace}
\newcommand{\primal}{\ensuremath{\mathcal{P}}\xspace}
\newcommand{\dual}{\ensuremath{\mathcal{D}}\xspace}
\newcommand{\wait}{\textsf{wait}}
\newcommand{\REAL}{\mathbb{R}}
\newcommand{\X}{\mathcal{X}}
\newcommand{\sur}{\textsf{sur}}
\newcommand{\sign}{\textsf{sgn}\xspace}
\newcommand{\act}{G}
\newcommand{\dv}[1]{\mathrm{d}#1}
\newcommand{\optcost}[1]{\textsf{opt-cost}(#1)} 
\newcommand{\x}[1]{x_{#1}} 
\newcommand{\y}[2]{y_{#1}(#2)}

\theoremstyle{plain}
\newtheorem{observation}[theorem]{Observation}

\title{A Primal-Dual Online Deterministic Algorithm for Matching with Delays}
 
\author{Marcin Bienkowski}{Institute of Computer Science, University of Wrocław, Poland}{marcin.bienkowski@cs.uni.wroc.pl}{https://orcid.org/0000-0002-2453-7772}{}
\author{Artur Kraska}{Institute of Computer Science, University of Wrocław, Poland}{artur.kraska@cs.uni.wroc.pl}{}{}
\author{Hsiang-Hsuan Liu}{Institute of Computer Science, University of Wrocław, Poland}{alison.hhliu@cs.uni.wroc.pl}{https://orcid.org/0000-0002-0194-9360}{}
\author{Pawe{\l} Schmidt}{Institute of Computer Science, University of Wrocław, Poland}{pawel.schmidt@cs.uni.wroc.pl}{}{}

\funding{Partially supported by Polish National Science Centre grant 2016/22/E/ST6/00499.}

\authorrunning{M. Bienkowski, A. Kraska, H. Liu and P. Schmidt}
\Copyright{Marcin Bienkowski, Artur Kraska, Hsiang-Hsuan Liu, Pawe{\l} Schmidt}

\subjclass{}

\ArticleNo{A}
\DOIPrefix{}
\nolinenumbers

\begin{document}

\maketitle

\begin{abstract}
In the Min-cost Perfect Matching with Delays (MPMD) problem, $2 m$ requests arrive
over time at points of a metric space. An online algorithm has to 
connect these requests in pairs, but a~decision to match may be 
postponed till a more suitable matching pair is found. The goal is to minimize the joint 
cost of connection and the total waiting time of~all~requests. 

We present an $O(m)$-competitive deterministic algorithm for this problem,
improving on an~existing bound of $O(m^{\log_2{5.5}}) = O(m^{2.46})$. Our
algorithm also solves (with the same competitive ratio) a bipartite variant of~MPMD, 
where requests are either positive or negative and only requests with
different polarities may be matched with each other. Unlike the existing
randomized solutions, our approach does not depend on~the~size of the metric
space and does not have to know it in advance.

\end{abstract}

\input{intro}
\input{lp}
\input{algo} 
\input{correctness}

\input{analysis-1}

\input{analysis-2}

\bibliographystyle{plainurl}
\bibliography{references}

\begin{appendix}
\input{tight}

\input{steiner}

\end{appendix}

\end{document}

%% file: intro.tex
\section{Introduction}

Consider a gaming platform that hosts two-player games, such as chess, go or
Scrabble, where participants are joining in real time, each wanting to play against
another human player. The system matches players according to their known
capabilities aiming at minimizing their dissimilarities: any player wants to compete
against an opponent with comparable skills. A better match for a~player can be found 
if the platform delays matching decisions as meanwhile more
appropriate opponents may join the system. However, an excessive delay may
also degrade the quality of experience. Therefore, a matching mechanism that
runs on a gaming platform has to balance two conflicting objectives: to
minimize the waiting time of any player and to minimize dissimilarities
between matched players.

The problem informally described above, called \emph{Min-cost Perfect Matching
with Delays (MPMD)}, has been recently introduced by Emek et
al.~\cite{EmKuWa16}. The problem is inherently online\footnote{The offline
variant of the problem, where all player arrivals are known a priori, can be
easily solved in polynomial time.}: a matching algorithm for the gaming platform
has to react in real time, without knowledge about future requests
(player arrivals) and make its decision irrevocably: once two requests
(players) are paired, they remain paired forever.

The MPMD problem was also considered in a \emph{bipartite} variant, called
\emph{Min-cost Bipartite Perfect Matching with Delays (MBPMD)} introduced by
Ashlagi et al.~\cite{AACCGK17}. There requests have polarities: one half of
them is positive, and the other half is negative. An algorithm may match only
requests of different signs. This setting corresponds to a variety of
real-life scenarios, e.g., assigning drivers to passengers on ride-sharing
platforms or matching patients to donors in kidney transplants. Similarly to
the MPMD problem, there is a trade-off between minimizing the waiting time and
finding a better match (a closer driver or a more compatible donor).


\subsection{Problem Definition} 

Formally, both in the MPMD and MBPMD problems, there is a metric space~$\X$
equipped with a distance function $\dist: \X \times \X \to \REAL_{\geq 0}$,
both known in advance to an online algorithm. An~online part of the input is a
sequence of $2 m$ requests $u_1, u_2, \ldots, u_{2m}$. A request 
(e.g.,~a~player arrival) $u$ is a triple $u = (\pos(u), \atime(u), \sign(u))$, where
$\atime(u)$ is the arrival time of request~$u$, $\pos(u) \in \X$ is the
request location, and $\sign(u)$ is the polarity of the request.

In the bipartite case, half of the requests are positive and $\sign(u) = +1$ for
any such request $u$; the remaining half are negative and there $\sign(u) = -1$.
In the non-bipartite case, requests do not have polarities, but for technical
convenience we set $\sign(u) = 0$ for any request $u$.

In applications described above, the function $\dist$ measures the
dissimilarity of a given pair of requests (e.g., discrepancy between player
capabilities in the gaming platform scenario or the physical distance between
a driver and a~passenger in the ride-sharing platform scenario). For instance,
for chess, a player is commonly characterized by her Elo rating
(an integer)~\cite{Elo78}. In such case, $\X$ may be simply a set of all
integers with the distance between two points defined as the difference of
their values.

Requests arrive over time, i.e., $\atime(u_1) \leq \atime(u_2) \leq \dots
\leq \atime(u_{2m})$. We note that the integer $m$ is not known beforehand to an
online algorithm. At any time $\tau$, an online algorithm may match a pair of
requests (players) $u$ and $v$ that 
\begin{itemize}
\item have already arrived ($\tau \geq \atime(u)$ and $\tau \geq \atime(v)$), 
\item have not been matched yet,
\item satisfy $\sign(u) = -\sign(v)$ (i.e., have opposite polarities in the bipartite case; in the 
non-bipartite case, this condition trivially holds for any pair).
\end{itemize}
The cost incurred by such \emph{matching edge} is then $\dist(\pos(u),\pos(v))
+ (\tau - \atime(u)) + (\tau - \atime(v))$. That is, it is the sum of
  the \emph{connection cost} defined as $\dist(\pos(u),\pos(v))$ and the
  \emph{waiting costs} of $u$ and $v$, defined as $\tau -
\atime(u)$ and $\tau - \atime(v)$, respectively.

The goal is to eventually match all requests and minimize the total cost of
all matching edges. We perform worst-case analysis, assuming that the
requests are given by an adversary. To measure the performance of an online
algorithm \ALGgeneric for an input instance $\I$, we compare its cost
$\ALGgeneric(\I)$ to the cost $\OPT(\I)$ of an optimal offline solution {\OPT}
that knows the entire input sequence in advance. The objective is to minimize
the competitive ratio~\cite{BorEl-98} defined as $\sup_\I \{ \ALGgeneric(\I) /
\OPT(\I)\}$.


\subsection{Previous Work}

The MPMD problem was introduced by Emek et al.~\cite{EmKuWa16}, who presented
a~\emph{randomized} $O(\log^2 n + \log \Delta)$-competitive algorithm. There,
$n$ is the number of points in the metric space $\X$ and $\Delta$ is its
aspect ratio (the ratio between the largest and the smallest distance
in~$\X$). The competitive ratio was subsequently improved by 
Azar~et~al.~\cite{AzChKa17} to $O(\log n)$. They also showed that the competitive
ratio of any randomized algorithm is at least $\Omega(\sqrt{\log n})$. The
currently best lower bound of $\Omega(\log n / \log \log n)$ for randomized
solutions was given by Ashlagi~et~al.~\cite{AACCGK17}.

Ashlagi et al.~\cite{AACCGK17} adapted the algorithm of Azar et
al.~\cite{AzChKa17} to the bipartite setting and obtained a randomized $O(\log
n)$-competitive algorithm for this variant. The currently best lower
bound of $\Omega(\sqrt {\log n / \log \log n})$ for this variant was also
given in \cite{AACCGK17}.

Both lower bounds use $O(n)$ requests. Therefore, they 
imply that no randomized algorithm can achieve a competitive ratio lower than
$\Omega(\log m / \log \log m)$ in the non-bipartite case and lower than 
$\Omega(\sqrt {\log m / \log \log m})$ in the bipartite one. (Recall that $2m$
is the number of requests in the input.)

The status of the achievable performance of \emph{deterministic} solutions
is far from being resolved. No better lower bounds than the ones used for
randomized settings are known for deterministic algorithms. 
The first solution that worked for general metric spaces was given by
Bienkowski et al.~and achieved an embarrassingly high competitive ratio of
$O(m^{\log_2{5.5}}) = O(m^{2.46})$~\cite{BiKrSc17}. Roughly speaking, their
algorithm is based on growing spheres around not-yet-paired
requests. Each sphere is created upon a request arrival, grows with time,
and when two spheres touch, the corresponding requests become matched.

Concurrently and independently of our current paper, Azar and
Jacob-Fanani~\cite{AzaJac18} improved the deterministic ratio to
$O((1/\varepsilon) \cdot m^{\log(3/2 + \varepsilon)})$, where $\varepsilon > 0$
is a parameter of their algorithm. When $\varepsilon$ is small enough, this
ratio becomes $O(m^{0.59})$. Their approach is similar to that
of~\cite{BiKrSc17}, but they grow spheres in a smarter way: slower than time
progresses and only in the negative direction of time axis.

Better deterministic algorithms are known only for simple spaces: 
Azar~et~al.~\cite{AzChKa17} gave an~$O(\textnormal{height})$-competitive algorithm for
trees and Emek et al.~\cite{EmShWa17} constructed a $3$-competitive
deterministic solution for two-point metrics (the latter competitive ratio is
best possible).


\subsection{Our Contribution}

In this paper, we focus on deterministic solutions for both the MPMD and MBPMD
problems, i.e., for both the non-bipartite and the bipartite variants of the
problem. We present a simple $O(m)$-competitive LP-based algorithm that works in
both settings. 

In contrast to the previous randomized solutions to these
problems~\cite{AACCGK17,AzChKa17,EmKuWa16}, and similarly to other
deterministic solutions~\cite{AzaJac18,BiKrSc17}, we do not need the metric
space $\X$ to be finite and known in advance by an online algorithm. (All
previous randomized solutions started by approximating $\X$ by a random HST
(hierarchically separated tree)~\cite{FaRaTa04} or a random HST tree with
reduced height~\cite{BaBuMN15}.) This approach, which can be performed only in
the randomized setting, greatly simplifies the task as the underlying tree
metric reveals a lot of structural information about the cuts between points
of $\X$ and hence about the structure of an optimal solution. In the
deterministic setting, such information has to be gradually learned as time
passes. For our algorithm, we require only that, together with any request
$u$, it learns the distances from $u$ to all previous requests.

In contrast to the previous deterministic algorithms~\cite{AzaJac18,BiKrSc17},
we base our algorithm on the moat-growing framework, developed originally for
(offline) constrained connectivity problems (e.g., for Steiner problems) by
Goemans and Williamson~\cite{GoeWil95}. Glossing over a lot of details, in
this framework, one writes a~primal linear relaxation of the problem and its
dual. The primal program has a constraint (connectivity requirement) for any
subset of requests and the dual program has a variable for any such subset.
The algorithm maintains a family of active sets, which are initially
singletons. In the runtime, dual variables are increased simultaneously, till
some dual constraint (corresponding to a pair of requests) becomes tight: in
such case an algorithm connects such pair and merges the corresponding sets.
At the end, an algorithm usually performs pruning by removing redundant edges.

When one tries to adapt the moat-growing framework to online setting, the main
difficulty stems from the irrevocability of the pairing decision: the pruning
operation performed at~the end is no longer an option. Another difficulty is
that an algorithm has to combine the concept of \emph{actual} time that passes
in an online instance with the \emph{virtual} time that dictates the growth
of dual variables. In particular, dual variables may only start to grow once
an online algorithm learns about the request and not from the very beginning
as they would do in the offline setting. Finally, requests appear online, and
hence both primal and dual programs evolve in time. For instance, this means
that for badly defined algorithms, appearing dual constraints may be violated
already once they are introduced.

We note that $2 m$ (the number of requests) is incomparable with $n$ (the number
of different points in the metric space $\X$) and their relation depends on
the application. Our algorithm is better suited for applications, where 
$\X$ is infinite or virtually infinite (e.g., it corresponds to an Euclidean plane
or a city map for ride-sharing platforms~\cite{LoVaJa16}) or very large (e.g.,
for some real-time online games, where player capabilities are represented as
multi-dimensional vectors describing their rank, reflex, offensive and
defensive skills, etc.~\cite{AvSpZa13}).


\subsection{Alternative Deterministic Approaches (That Fail)}

A few standard deterministic approaches fail when applied to 
the MPMD and MBPMD problems. One such attempt is the \emph{doubling technique} (see,
e.g.,~\cite{ChrKen06}):  an online algorithm may trace the cost of an optimal
solution \OPT and perform a global operation (e.g., match many pending
requests) once the cost of \OPT increases significantly (e.g., by a factor of
two) since the last time when such global operation was performed. 
This approach does not seem to be feasible here as the total cost of \OPT may
\emph{decrease} when new requests appear.

Another attempt is to observe that the randomized algorithm by 
Azar~et~al.~\cite{AzChKa17} is a~deterministic algorithm run on a random tree
that approximates the original metric space.  One may try to replace a random
tree by a deterministically generated tree that spans requested points of the metric 
space. Such spanning tree can be computed by 
the standard greedy routine for the online Steiner tree problem~\cite{ImaWax91}.
However, it turns out that the competitive ratio of the resulting algorithm
is $2^{\Omega(m)}$.  (The main reason is that the adversary may give an
initial subsequence that forces the algorithm to create a spanning tree with
the worst-case stretch of $2^{\Omega(m)}$ and such initial subsequence can be served
by \OPT with a~negligible cost. The details are given in
Appendix~\ref{sec:steiner}.)


\subsection{Related Work}

Originally, online metric matching problems have been studied in variants where
delaying decisions was not permitted. In this variant,
$m$ requests with positive polarities are given at the beginning to
an algorithm. Afterwards, $m$ requests with negative polarities are
presented one by one to an algorithm and they have to be matched
\emph{immediately} to existing positive requests. The goal is
to minimize the weight of a perfect matching created by the algorithm. For
general metric spaces, the best deterministic
algorithms achieve the optimal competitive ratio of $2m-1$~\cite{KalPru93,KhMiVa94,Raghve16}
and the best randomized solution is $O(\log^2 m)$-competitive~\cite{BaBuGN14,MeNaPo06}.
Better bounds are known for line metrics~\cite{AnBNPS14,FuHoKe05,GupLew12,KouNan03}: here the 
best deterministic algorithm is $O(\log^2 m)$-competitive~\cite{NayRag17}
and the best randomized one achieves the ratio of $O(\log m)$~\cite{GupLew12}.

Another strand of research concerning online matching problems arose around a
non-metric setting where points with different polarities are connected by
graph edges and the goal is to maximize the cardinality or the weight of the
produced matching. For a comprehensive overview of these type of problems we
refer the reader to a recent survey by Mehta~\cite{Meht13}.

The M(B)PMD problem is an instance in a broader category of problems, where an
online algorithm may delay its decisions, but such delays come with a~certain
cost. Similar trade-offs were employed in other areas of online analysis: in
aggregating orders in supply-chain management
\cite{BBBCDF16,BBCJNS14,BBCJSS13,BuFeNT17,BuKLMS08}, aggregating messages in
computer networks~\cite{DoGoSc01,KaKeRa03,KhNaRa02}, or recently for server
problems~\cite{AzGaGP17,BiKrSc18}.

%% file: lp.tex
\section{Primal-Dual Formulation}

We start with introducing a linear program that allows us to lower-bound the
cost of an~optimal solution. To this end, fix an instance \I of M(B)PMD. Let
\R be the set of all requests. We call any unordered pair of different
requests in \I an edge; let \E be the set of all edges that correspond to
potential matching pairs, i.e., the set of all edges in the non-bipartite
case, and the edges that connect requests of opposite polarities in the
bipartite variant. For each set $S\subseteq \R$, by~$\delta(S)$ we denote the
set of all edges from \E crossing the boundary of $S$, i.e., having exactly
one endpoint in $S$.

For any set $S\subseteq \R$, we define $\sur(S)$ (\emph{surplus} of set $S$)
as the number of unmatched requests in a maximum cardinality matching of
requests within set $S$. 
\begin{itemize}
\item
In the non-bipartite variant (MPMD), we are allowed to match any two requests.
Hence, if $S$ is of even size, then $\sur(S) = 0$. Otherwise, $\sur(S) = 1$ as
in any maximum cardinality matching of requests within $S$ exactly one request
remains unmatched.
\item
In the bipartite variant (MBPMD), we can always match two requests of
different polarities. Hence, the surplus of a set $S$ is the discrepancy
between the number of positive and negative requests inside $S$, i.e.,
$\sur(S) = |\sum_{u\in S} \sign(u)|$.
\end{itemize}

To describe a matching, we use the following notation. For each edge $e$, we
introduce a binary variable $x_e$, such that $x_e = 1$ if and only if $e$ is a
matching edge. For any set $S\subseteq \R$ and any feasible matching (in
particular the optimal one), it holds that $\sum_{e \in \delta(S)} x_e \geq
\sur(S)$.

Fix an optimal solution \OPT for \I. If a pair of requests $e = (u,v)$ is
matched by \OPT, it is matched as soon as both $u$ and $v$
arrive, and hence the cost of matching $u$ with $v$ in the solution of \OPT is
equal to $\optcost{e} := \dist(\pos(u),\pos(v)) + |\atime{(u)} -
\atime{(v)}|$. This, together with the preceding observations, motivates the
following linear program \primal:
\begin{align*}
	\text{minimize \quad} & \sum_{e \in \E} \optcost{e} \cdot \x{e} & \\
	\text{subject to \quad} & \sum_{e \in \delta(S)} \x{e} \geq \sur(S) & \quad \forall{S \subseteq \R} \\
	& \x{e} \geq 0 & \quad \forall e \in \E.
\end{align*}

As any matching is a feasible solution to \primal, the cost of the optimal
solution of \primal lower-bounds the cost of the optimal solution for instance
\I of M(B)PMD. Note that there might exist a feasible integral solution of
\primal that does not correspond to any matching. To exclude all such
solutions, we could add constraints $\sum_{e \in \delta(S)} \x{e} = 1$ for all
singleton sets $S$. The resulting linear program would then exactly  describe
the matching problem (cf. Chapter 25 of \cite{Schrij03}). However, our main
concern is not \primal, but its dual and its current shape is sufficient for
our purposes. The program \dual, dual to \primal, is then
\begin{align*}
	\text{maximize \quad} & \sum_{S\subseteq \R} \sur(S) \cdot y_S & \\
	\text{subject to \quad} & \sum_{S:e \in \delta(S)} y_S \leq \optcost{e} & \quad \forall{e \in \E} \\
	& y_S \geq 0  & \quad \forall{S}\subseteq \R.
\end{align*}

Note that in any solution, the dual variables $y_S$ corresponding to sets $S$ for which $\sur(S)=0$, can be set to $0$ without changing feasibility or objective value.

The following lemma is an immediate consequence of weak duality. 

\begin{lemma}
\label{lem:dual}
	Fix any instance \I of the M(B)PMD problem. Let $\OPT(\I)$ be the value of
	any optimal solution of \I and $D$ be the value of any feasible solution
	of~$\dual$. Then $\OPT(\I) \geq D$.
\end{lemma}

\begin{proof}
	Let $P^*$ and $D^*$ be the values of optimal solutions for $\primal$ and
	$\dual$, respectively. Since any matching is a feasible solution for
	$\primal$, $\OPT(\I) \geq P^*$. Hence, $\OPT(\I) \geq P^* \geq D^* \geq
	D$.
\end{proof}

Lemma~\ref{lem:dual} motivates the following approach: We construct an
online algorithm \ALGlong (\ALG), which, along with its own
solution, maintains a feasible solution $D$ for \dual corresponding to the
already seen part of the input instance. This feasible dual solution not only
yields a lower bound on the cost of the optimal matching, but also plays a
crucial role in deciding which pair of requests should be matched.

Note that since the requests arrive in an online manner, \dual evolves in
time. When a~request arrives, the number of subsets of \R increases (more
precisely, it doubles), and hence more dual variables $y_S$ are introduced.
Moreover, the newly arrived request creates an edge with every existing request
and the corresponding dual constraints are introduced. Therefore, showing the
feasibility of the created dual solution is not immediate; we deal with
this issue in Section~\ref{sec:correctness}.

%% file: algo.tex
\section{Algorithm \ALGlong}

The high-level idea of our algorithm is as follows: \ALGlong (\ALG) resembles
moat-growing algorithms for solving constrained forest
problems~\cite{GoeWil95}. During its runtime, \ALG partitions all the requests
that have already arrived into \emph{active sets}.\footnote{A reader familiar
with the moat-growing algorithm may think that active sets are moats. However,
not all of them are growing in time.} If an active set contains any free
requests, we call this set \emph{growing}, and \emph{non-growing} otherwise. At
any time, for each active growing set $S$, the algorithm increases
continuously its dual variable $y_S$ until a constraint in \dual corresponding to
some edge $(u,v)$ becomes tight. When it happens, \ALG makes both active sets
(containing $u$ and $v$, respectively) inactive, and the set being their union
active. In addition, if this happened due to two growing sets,
\ALG matches as many pairs of free requests in these sets as possible: in the 
non-bipartite variant \ALG matches exactly one pair of free requests, while in
the bipartite variant, \ALG matches free requests of different polarities
until all remaining free requests have the same sign.


\subsection{Algorithm Description}

More precisely, at any time, \ALG partitions all requests that arrived until
that time into \emph{active} sets. It maintains mapping $\Moat$, which assigns
an active set to each such request. An~active set $S$, whose all requests are
matched is called \emph{non-growing}. Conversely, an active set $S$ is called
\emph{growing} if it contains at least one free request. \ALG ensures that the
number of free requests in an active set $S$ is always equal to $\sur(S)$. We
denote the set of free requests in an~active set $S$ by $\Free(S)$; if $S$ is
non-growing, then $\Free(S) = \emptyset$.

When a request $u$ arrives, the singleton $\{u\}$ becomes a new active and
growing set, i.e., $\Moat(u) = \{u\}$. The dual variables of all active
growing sets are increased continuously with the same rate in which time
passes. This increase takes place until a dual constraint between two
active sets becomes tight, i.e., until there exists at least one edge $e =
(u,v)$, such that
\begin{equation}
	\label{eq:cutEdge}
	 \Moat(u) \neq \Moat(v) \quad\text{ and } \sum_{S: e \in \delta(S)} y_S = \optcost{e}.
\end{equation}
In such case, while there exists an edge $e = (u,v)$ satisfying
\eqref{eq:cutEdge},
\ALG processes such edge in the following way. First, it \emph{merges} active 
sets $\Moat(u)$ and $\Moat(v)$. By merging we mean that the mapping \Moat is
adjusted to the new active set $S = \Moat(u)
\uplus \Moat(v)$ for each request of $S$. Old active sets $\Moat(u)$ and
$\Moat(v)$ become \emph{inactive}.\footnote{Note that \emph{inactive} is not
the opposite of being \emph{active}, but means that the set was active
previously: some sets are never active or inactive.} Second, as long as
there is a pair of free requests $u', v' \in S$ that can be matched with each
other, \ALG matches them.

In the non-bipartite variant, \ALG matches at most one pair as each
active set contains at most one free request. In the bipartite variant, \ALG
matches pairs of free requests until all unmatched requests in $S$ (possibly zero) have the
same polarity. Observe that in either case, the number of free requests
after merge is equal to $\sur(S)$. Finally, \ALG \emph{marks} edge $e$. Marked
edges are used in the analysis, to find a proper charging of the
connection cost to the cost of the produced solution for
\dual. The pseudocode of \ALG is given in Algorithm~\ref{alg:the-algorithm2}
and an example execution that shows a partition of requests into active sets is
given in Figure~\ref{fig:clustering}.


\algdef{SE}[SUBALG]{Indent}{EndIndent}{}{\algorithmicend\ }%
\algtext*{Indent}
\algtext*{EndIndent}
\begin{algorithm}[t]
	\caption{Algorithm \ALGlong}
		\label{alg:the-algorithm2}
	\begin{algorithmic}[1]
		\State{\textbf{Request arrival event:}} 
			\Indent
			\If{a request $u$ arrives}
				\State{$\Moat(u) \gets \{u\}$}
				\ForAll{sets $S$ such that $u\in S$}
					\State{$y_S \gets 0$}
					\Comment{\emph{initialize dual variables for sets containing $u$}} 
				\EndFor
			\EndIf
			\EndIndent
		\State{}
		\State{\textbf{Tight constraint event:}}
		\Indent
			\While{exists a tight dual constraint for edge $e = (u,v)$ where $\Moat(u) \neq \Moat(v)$}
				\label{alg:tightConstraint}
				\State{$S\gets \Moat(u) \uplus \Moat(v)$ } 
					\Comment{\emph{merge two active sets}}
				\ForAll{$w\in S$}	
					\Comment{\emph{adjust assignment $\Moat$ for the new active set $S$}}
					\State{$\Moat(w) \gets S$}
				\EndFor	
				\State{\textbf{mark} edge $e$}
				\label{alg:markEdge}
				\While{there are $u', v' \in \Free(S)$ such that $\sign(u') = -\sign(v')$}
					\State{\textbf{match} $u'$ with $v'$}
						\Comment{\emph{match as many pairs as possible}}
				\EndWhile
			\EndWhile
		\EndIndent
		\State{}
		\State{\textbf{None of the above events occurs:}}
		\Indent
			\ForAll{growing active sets $S$}
				\State{increase continuously $y_S$ with the same rate in which time passes}
			\EndFor
		\EndIndent
	\end{algorithmic}
\end{algorithm}


\begin{figure}[t]
	\centering
	\includegraphics[width=\textwidth]{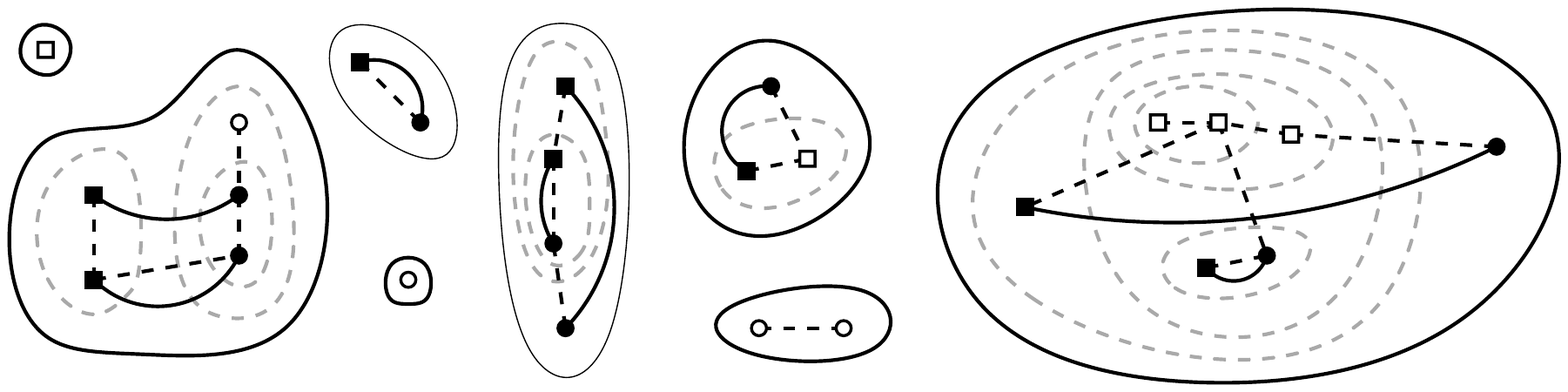}
	\caption{A partition of requests into active sets created by \ALG.
	Different polarities of requests are represented by discs and squares. Free
	requests are depicted as empty discs and squares, matched requests by
	filled ones. Active growing sets have bold boundaries and each of them
	contains at least one free request. Active non-growing sets contain only matched
	requests. Dashed lines represent marked edges and solid curvy lines
	represent matching edges. Dashed gray sets are already inactive; the
	inactive singleton sets have been omitted.}
	\label{fig:clustering} 
\end{figure}


\subsection{\ALGlong Properties}

It is instructive to trace how the set $\Moat(u)$ changes in time for a
request $u$. At the beginning, when $u$ arrives, $\Moat(u)$ is just the
singleton set $\{u\}$. Then, the set~$\Moat(u)$ is merged at least once with
another active set. If $\Moat(u)$ is merged with a non-growing set, the number
of requests in $\Moat(u)$ increases, but its surplus remains intact. After
$\Moat(u)$ is merged with a growing set, some requests inside the new
$\Moat(u)$ may become matched. It is possible that, in effect, the surplus of
the new set $\Moat(u)$ is zero, in which case the new set $\Moat(u)$ is non-growing.
(In the non-bipartite variant, this is always the case when two growing sets
merge.) After $\Moat(u)$ becomes non-growing, another growing set may
be merged with $\Moat(u)$, and so on. Thus, the set $\Moat(u)$ can change its
state from growing to non-growing (and back) multiple times.

The next observation summarizes the process described above, listing
properties of \ALG that we use later in our proofs.

\begin{observation}
\label{obs:alg_properties}
The following properties hold during the runtime of \ALG.
\begin{enumerate}
	\item 
		For a request $u$, when time passes, $\Moat(u)$ refers to different
		active sets that contain $u$.
	\item
		\label{item:one_to_one}
		At any time, every request is contained in exactly one active set. If this request
		is free, then the active set is growing.
	\item 
		\label{item:friseur}
		At any time, an active set $S$ contains exactly $\sur(S)$ free requests.
	\item
		\label{item:laminar}
		Active and inactive sets together constitute a laminar family of sets.
	\item 
		\label{item:once-in-love-always-in-love}
		For any two requests $u$ and $v$, once $\Moat(u)$ becomes equal to
		$\Moat(v)$, they will be equal forever.
\end{enumerate}
\end{observation}

%% file: correctness.tex
\section{Correctness} 
\label{sec:correctness}

We now prove that \ALGlong is defined properly. In other words, we show that
the dual values maintained by $\ALG$ always form a feasible solution
of $\dual$ (Lemma~\ref{lem:feasible_dual}) and $\ALG$ returns a feasible
matching of all requests at the end (Lemma~\ref{lem:proper_matched}). From
now on, we denote the values of a dual variable $y_S$ at time $\tau$ by
$y_S(\tau)$.

By the definition, the waiting cost of a request is the time difference
between the time it arrives and the time it is matched. In the following
lemma, we relate the waiting cost of a request to the dual variables for the
active sets it belongs to.

\begin{lemma}
\label{lem:request_wc_ub}
Fix any request $u$. For any time $\tau \geq \atime(u)$, it holds that
$$\sum_{S:u\in S} y_S(\tau) \leq \tau-\atime(u).$$ The relation holds with
equality if $u$ is free at time $\tau$.
\end{lemma}

\begin{proof}
We show that the inequality is preserved as time passes. At time $\tau =
\atime(u)$, request $u$ is introduced and sets $S$ containing $u$ appear.
Their $y_S$ values are initialized to $0$. Therefore, at that time,
$\sum_{S:u\in S} y_S(\tau) = 0$ as desired.

Whenever a merging event or an arrival of any other requests occur, new
variables $y_S$ may appear in the sum $\sum_{S:u\in S} y_S(\tau)$, but, at
these times, the values of these variables are equal to zero, and therefore do
not change the sum value.

It remains to analyze the case when time passes infinitesimally by
$\varepsilon$ and no event occurs within this period. It is sufficient to
argue that the sum $\sum_{S:u\in S} y_S(\tau)$ increases exactly by
$\varepsilon$ if $u$ is free at $\tau$ and at most by $\varepsilon$ otherwise.
Recall that $y_S$ may grow only if $S$ is an active growing set. By
Property~\ref{item:one_to_one} of Observation~\ref{obs:alg_properties},
the only active set containing $u$ is $\Moat(u)$. This set is growing if $u$
is free (and then $y_{\Moat(u)}$ increases exactly by $\varepsilon$) and may
be growing or non-growing if $u$ is matched (and then $y_{\Moat(u)}$ increases
at most by $\varepsilon$).
\end{proof}

The following lemma shows that throughout its runtime, $\ALG$ maintains 
a~feasible dual solution.

\begin{lemma}
\label{lem:feasible_dual}
At any time, the values $y_S$ maintained by the algorithm constitute 
a~feasible solution to \dual.
\end{lemma}

\begin{proof}
We show that no dual constraint is ever violated during the execution of~\ALG.

When a new request $u$ arrives at time $\tau=\atime(u)$, new sets
containing $u$ appear and the dual variables $y_S$ corresponding to these sets
are initialized to $0$. 

Each already existing constraint, corresponding to an edge $e$ not incident to~$u$, 
is modified: new $y_S$ variables for sets $S$ containing both $u$ and
exactly one of endpoints of $e$ appear in the sum. However, all these variables
are zero, and hence the feasibility of such constraints is preserved.

Moreover, for any edge $e = (u, v)$ where $v$ is an
existing request, a new dual constraint for this edge appears in \dual. We
show that it is not violated, i.e., $\sum_{S: e \in \delta(S)}
y_S(\tau) \leq \optcost{e}$. As discussed before, $y_S(\tau) = 0$ for the
sets $S$ containing $u$. Therefore,
\begin{align*} 
	\sum_{S: e \in \delta(S)} y_S(\tau) 
		& = \sum_{S: v\in S \wedge u\notin S} y_S(\tau) + \sum_{S: u\in S \wedge v\notin S} y_S(\tau) \\
		& = \sum_{S: v\in S \wedge u\notin S} y_S(\tau) 
		 \leq \sum_{S: v\in S} y_S(\tau) \\
		& \leq \atime(u) - \atime(v) & \text{(by Lemma~\ref{lem:request_wc_ub})}\\
		& \leq \optcost{e}.
\end{align*}

Now, we prove that once a dual constraint for an edge $e = (u,v)$ becomes
tight, the involved $y_S$ values are no longer increased. According to the
algorithm definition, $\Moat(u)$ and~$\Moat(v)$ become merged together. By
Property~\ref{item:once-in-love-always-in-love} of
Observation~\ref{obs:alg_properties}, from this moment on, any active
set $S$ contains either both $u$ and~$v$ or neither of them. Hence, there is
no active set $S$, such that $(u,v)\in \delta(S)$, and in particular there is
no such active growing set. Therefore, the value of $\sum_{S: e \in
\delta(S)} y_S$ remains unchanged, and hence the dual constraint
corresponding to edge $e$ remains tight and not violated.
\end{proof}

Finally, we prove that $\ALG$ returns a proper matching. 
We need to show that if a pair of requests remains unmatched, then 
appropriate dual variables increase and they will eventually trigger the 
matching event. 

\begin{lemma}
	\label{lem:proper_matched}
	For any input for the M(B)PMD problem, $\ALG$ returns a feasible matching.
\end{lemma}

\begin{proof}
	Suppose for a contradiction that \ALG does not match some request $u$.
	Then, by~Property~\ref{item:one_to_one} of
	Observation\ref{obs:alg_properties}, $\Moat(u)$ is always an active
	growing set and by Property~\ref{item:friseur}, $\sur(\Moat(u))>0$.
	Therefore, the corresponding dual variable $y_{\Moat(u)}$ always increases
	during the execution of \ALG and appears in the objective function of
	\dual with a positive coefficient. By~Lemma~\ref{lem:feasible_dual}, the
	solution of \dual maintained by \ALG is always feasible, and hence the
	optimal value of \dual would be unbounded. This would be a contradiction,
	as there exists a finite solution to the primal program \primal (as all
	distances in the metric space are finite).
\end{proof}

%% file: analysis-1.tex
\section{Cost Analysis}

In this section, we show how to relate the cost of the matching returned
by \ALGlong to the value of the produced dual solution. First, we
show that the total waiting cost of the algorithm is equal to the value
of the dual solution. Afterwards, we bound the connection cost of \ALG
by $2 m$ times the dual solution, where $2 m$ is the number of requests in the
input. This, along with Lemma~\ref{lem:dual}, yields the competitive
ratio of $2m+1$.

\subsection{Waiting Cost}

In the proof below, we link the generated waiting cost with the growth of 
appropriate dual variables. To this end, suppose that a set $S$ is an active
set for time period of length $\Delta t$. By~Property~\ref{item:friseur} of
Observation~\ref{obs:alg_properties}, $S$ contains exactly $\sur(S)$ free
points, and thus the waiting cost incurred within this time by requests in $S$
is $\Delta t \cdot \sur(S)$. Moreover, in the same time interval, the dual
variable $y_S$ increases by $\Delta t$, which contributes the same amount,
$\sur(S) \cdot \Delta t$, to the growth of the objective function of $\dual$.
The following lemma formalizes this observation and applies it to all active
sets considered by \ALG in its runtime.

\begin{lemma}
	\label{lem:wc}
	The total waiting cost of \ALG is equal to $\sum_{S \subseteq \R} \sur(S)
	\cdot \y{S}{T}$, where $T$ is the time when \ALG matches the last request.
\end{lemma} 

\begin{proof}
We define $\act(\tau)$ as the family of sets that are active and growing at
time~$\tau$. By~Property~\ref{item:one_to_one} and
Property~\ref{item:friseur} of Observation~\ref{obs:alg_properties}, the
number of free requests at time $\tau$, henceforth denoted $\wait(\tau)$, is
then equal to $\sum_S \sur(S) \cdot \mathds{1}[S \in \act(\tau)]$. The total
waiting cost at time $T$ can be then expressed as
\begin{align*}
	\int_0^T \wait(\tau) \,\dv{\tau} 
		& = \int_0^T \sum_S \sur(S) \cdot \mathds{1}[S \in \act(\tau)] \,\dv{\tau} \\
		& = \sum_S \sur(S) \int_0^T \mathds{1}[S \in \act(\tau)] \,\dv{\tau}
		= \sum_S \sur(S) \cdot y_S(T),
\end{align*}
where the last equality holds as at any time, \ALG increases $y_S$ value if
and only if $S$ is active and growing.
\end{proof}

%% file: analysis-2.tex
\subsection{Connection Cost}
\label{sec:connCost}

Below, we relate the connection cost of \ALG to the value of the final
solution of \dual, created by \ALG. We focus on the set of marked edges, which
are created by \ALG in Line~\ref{alg:markEdge} of
Algorithm~\ref{alg:the-algorithm2}. We show that for any time, the set of
marked edges restricted to an active or an inactive set~$S$ forms a ``spanning
tree'' of requests of $S$. That is, there is a unique path of marked edges
between any two requests from $S$. (Note that this path projected to the
metric space may contain cycles as two requests may be given at the same point
of $\X$.) We start with a~helper observation.

\begin{observation}
	\label{obs:marked_boundary}
	Fix any set $S$. If $S$ is active at time $\tau$, then its boundary $\delta(S)$ 
	does not contain any marked edge at time $\tau$.
\end{observation}

\begin{proof}
	After an edge $(u, v)$ becomes marked, both $u$ and $v$ belong to newly
	created active set. From now on, by
	Property~\ref{item:once-in-love-always-in-love} of
	Observation~\ref{obs:alg_properties}, they remain in the same active set
	till the end of the execution. Therefore, this edge will never be
	contained in a boundary of an~active set.
\end{proof}

\begin{lemma}
	\label{lem:components}
	At any time, for any active or inactive set $S$, the subset of all marked edges 
	with both endpoints in $S$ forms a spanning tree of all requests from $S$.
\end{lemma}

\begin{proof}
	We show that the property holds at time passes. When a new request
	arrives, a new active growing set containing only one request is created.
	This set is trivially spanned by an empty set of marked edges.
		
	By the definition of \ALG, a new active set appears when a dual constraint
	for some edge $e = (u,v)$ becomes tight. Right before it happens, the
	active sets containing $u$ and $v$ are $\Moat(u)$ and $\Moat(v)$,
	respectively. At that time, marked edges form spanning trees of sets
	$\Moat(u)$ and $\Moat(v)$ and, by Observation~\ref{obs:marked_boundary},
	there are no marked edges between these two sets. Hence, these spanning
	trees together with the newly marked edge $e$ constitute a spanning tree
	of the requests of $S = \Moat(u) \uplus \Moat(v)$. Finally, a set may
	become inactive only if it was active before, and \ALG never adds any
	marked edge inside an already existing active or inactive set.
\end{proof}

Using the lemma above, we are ready to bound the connection cost of one matching edge by 
the cost of the solution of \dual. 

\begin{lemma}
	\label{lem:cc}
	The connection cost of any matching edge is at most $2 \cdot
	\sum_{S\subseteq \R} \sur(S)\cdot y_S(T)$, where $T$ is the time when \ALG
	matches the last request.
\end{lemma} 

\begin{proof}
	Fix a matching edge $(u, v)$ created by \ALG at time $\tau$. Its
	connection cost is the distance $\dist(\pos(u),\pos(v))$ between the
	points corresponding to requests $u$ and $v$ in the underlying metric
	space.

	We consider the state of \ALG right after it matches $u$ with $v$. By
	Lemma~\ref{lem:components}, the active set $S = \Moat(u) = \Moat(v)$
	containing $u$ and $v$ is spanned by a tree of marked edges. Let $P$ be
	the (unique) path in this tree connecting $u$ with $v$. Using the triangle
	inequality, we can bound $\dist(\pos(u),\pos(v))$ by the length of $P$
	projected onto the underlying metric space.

	Recall that for any edge $e = (w,w')$, it holds that
	$\dist(\pos(w),\pos(w')) \leq \optcost{e}$. Moreover, if $e$ is marked,
	the dual constraint for edge $e$ holds with equality, that is,
	$\optcost{e} = \sum_{S: e\in \delta(S)} y_S(\tau)$. Therefore,
	\begin{align*} 
		\dist(\pos(u),\pos(v)) 
			& \leq \sum_{(w,w') \in P} \dist(\pos(w), \pos(w')) 
			 \leq \sum_{e \in P} \optcost{e} \\ 
			& = \sum_{e \in P} \sum_{S:e \in \delta(S)} y_S(\tau) 
			 = \sum_{S} |\delta(S) \cap P| \cdot y_S(\tau) \\
			& \leq \sum_{S} |\delta(S) \cap P| \cdot \sur(S) \cdot y_S(\tau) \\
			& \leq \sum_{S} |\delta(S) \cap P| \cdot \sur(S) \cdot y_S(T).
	\end{align*}
	The penultimate inequality holds because a dual variable $y_S$ can be
	positive only if $\sur(S) \geq 1$. It is now sufficient to prove that for
	each (active or inactive) set~$S$, it holds that $|\delta(S) \cap P| \leq 2$, i.e., the
	path $P$ crosses each such set $S$ at most twice.
	
	For a contradiction, suppose that there exists an (active or inactive) set
	$S$, whose boundary is crossed by path $P$ more than twice. We direct all edges
	on~$P$ towards $v$ (we follow $P$ starting from request $u$ and move
	towards $v$). Note that $u$ may be inside or outside of $S$.
	Let $e_1 = (w_1, w_2)$ be the first edge on $P$ such that $w_1\in
	S$ and $w_2\not\in S$, i.e., the first time when path $P$ leaves $S$.
	Let $e_2 = (w_3, w_4)\in P$ be the first edge after $e_1$, 
	such that $w_3\not \in S$ and $w_4 \in S$, that is, the first time when path
	$P$ returns to $S$ after leaving it with edge $e_1$.
	Edge $e_2$ must exist as we assumed that $P$ crosses the boundary of 
	$S$ at least three times.

	By Lemma~\ref{lem:components}, a subset of the marked edges constitutes a
	spanning tree of $S$. Hence, there exists a path of marked edges contained
	entirely in $S$ that connects requests $w_1$ and $w_4$. Furthermore, a
	sub-path of $P$ connects $w_2$ and $w_3$ outside of~$S$. These two paths
	together with edges $e_1$ and $e_2$ form a cycle of marked edges. However,
	by Lemma~\ref{lem:components} and
	Observation~\ref{obs:marked_boundary}, at any time, the set of marked edges
	forms a forest, which is a contradiction.
\end{proof}

\subsection{Bounding the Competitive Ratio}

Using above results we are able to bound the cost of \ALGlong.

\begin{theorem}
	\label{thm:main}
		\ALGlong is $(2m+1)$-competitive for the M(B)PMD problem.
	\end{theorem}  

\begin{proof}
	Fix any input instance $\I$ and let $\dual$ be the corresponding dual
	program. Let $D$ be the cost of the solution to \dual output by \ALG. By
	Lemma~\ref{lem:wc}, the total waiting cost of the algorithm is bounded
	by $D$ and by Lemma~\ref{lem:cc}, the connection cost of a single edge
	in the matching is bounded by $2\cdot D$. Therefore,
	\[
		\ALG(\I) \leq D + m \cdot 2 D = (2m+1) \cdot D \leq (2m+1) \cdot \OPT(\I),
	\]
	where the first inequality holds as there are exactly $m$ matched edges
	and the last equality follows by Lemma~\ref{lem:dual}.
\end{proof}

%% file: tight.tex
\section{Tightness of the Analysis}

We can show that our analysis of \ALGlong is asymptotically tight, i.e., the
competitive ratio of \ALGlong is $\Omega(m)$.

\begin{theorem}
Both for MPMD and MBPMD problems, there exists an instance~\I, such that 
$\ALG(\I) = \Omega(m) \cdot \OPT(\I)$.
\end{theorem}

\begin{proof}
Let $m>0$ be an even integer and $\varepsilon = 1/m$. Let $\mathcal{X}$ be
the metric containing two points $p$ and $q$ at distance $2$.
	
In the instance \I, requests are released at both points $p$ and $q$ at
times $0, 1+\varepsilon, 1+3 \varepsilon, 1+5 \varepsilon, \ldots, 1+(2m-3)
\cdot \varepsilon$. For the MBPMD problem, we additionally specify request
polarities: at $p$, all odd-numbered requests are positive and all
even-numbered are negative, while requests issued at $q$ have exactly opposite
polarities from those at $p$.
	
Regardless of the variant (bipartite or non-bipartite) we solve, \ALG matches
the first pair of requests at time $1$, when their active growing sets are
merged, forming a new active non-growing set. Every subsequent pair of
requests appears exactly $\varepsilon$ after the previous pair becomes matched.
Therefore, they are matched together $\varepsilon$ after their arrival, when
their growing sets are merged with the large non-growing set containing all
the previous pairs of requests. Hence, the total connection cost of \ALG is
equal to $2 m$. On the other hand, observe that the total cost of a solution
that matches consecutive requests at each point of the metric space separately
is equal to $2 \cdot ((1 + \varepsilon) + 2 \varepsilon \cdot (m-2)/2) = 2
\cdot (1 + (m-1) \cdot \varepsilon) < 4$. 
\end{proof}

%% file: steiner.tex
\section{Derandomization Using a Spanning Tree}
\label{sec:steiner}

In this part, we analyze an algorithm that approximates the metric space
by a greedily and deterministically chosen spanning tree of requested points 
and employs the deterministic algorithm for trees of Azar et al.~\cite{AzChKa17}.
We show that such algorithm has the competitive ratio of $2^{\Omega(m)}$.
For simplicity, we focus on the non-bipartite variant, but the lower bound can be 
easily extended to the bipartite case.

More precisely, we define a natural algorithm \textsc{Tree Based} (\TB).
\TB internally maintains a spanning tree $T$ of 
metric space points corresponding to already seen requests. That is, whenever
\TB receives a request $u$ at point $\pos(u)$, it executes the following two
steps.
\begin{enumerate}
\item If there was no previous request at $\pos(u)$, \TB adds $\pos(u)$ to $T$, 
connecting it to the closest point from $T$. The addition is performed immediately,
at the request arrival. This part essentially mimics the behavior of the greedy algorithm 
for the online Steiner tree problem~\cite{ImaWax91}. 

\item To serve the request $u$, \TB runs the deterministic algorithm of~\cite{AzChKa17} on
the tree $T$.\footnote{The algorithm must be able to operate on a tree that may be
extended (new leaves may appear) in the runtime. The algorithm given by Azar
et al.~\cite{AzChKa17} has this property.} 
\end{enumerate}

\begin{theorem}	
The competitive ratio of \textsc{Tree Based} is $2^{\Omega(m)}$.
\end{theorem}

\begin{proof} 
The idea of the lower bound is as follows. The adversary first gives $m/2$
requests that force \TB to create a tree $T$ with the stretch of
$2^{\Omega(m)}$ and then gives another $m/2$ requests, so that the initial $m$
requests can be served with a~negligible cost by $\OPT$. Afterwards, the
adversary consecutively requests a pair of points that are close in the metric
space, but  far away in the tree $T$.

Our metric space $\X$ is a continuous ring and we assume that $m$ is an~even
integer. Let $h$~be the length of this ring and let $\varepsilon = h/(m \cdot
2^{m-1})$.

In the first part of the input, the adversary gives $m/2$ requests in the
following way. The first two requests are given at time $0$ at antipodal
points (their distance is $h/2$). \TB connects them using one of two halves of
the ring. From now on, the tree $T$ of \TB will always cover a~contiguous part
of the ring. Each of the next $m/2 - 2$ requests is given exactly in the
middle of the ring part not covered by~$T$. For $j \in \{3, 4, \ldots, m/2
\}$, the $j$-th request is given at time $(2 \cdot (j-1) / m) \cdot
\varepsilon$.

This way, the ring part not covered by $T$ shrinks exponentially, and after
$m/2$ initial requests its length is equal to $h/2^{m/2-1}$. Let $p$ and $q$
be the endpoints (the only leaves) of~$T$. Then, $\dist(p,q) = h/2^{m/2-1}$,
but the path between $p$ and~$q$ in~$T$ is of length $h - \dist(p,q)$ and uses
an edge of length $h/2$.  As $T$ is built as soon as requests appear, its
construction is finished right after the appearance of the $(m/2)$-th request,
i.e., before time $\varepsilon$.

In the second part of the input, at time $\varepsilon$, the adversary gives
$m/2$ requests at the same points as the requests from the first phase. This
way, \OPT may serve the first $m$ requests paying nothing for the connection
cost and paying at most $(m / 2) \cdot \varepsilon = h/2^m$ for their waiting
cost.

In the third part of the input, the adversary gives $m/2$ pairs of requests,
each pair at points $p$ and $q$.  Each pair is given after  the previous
one is served by \TB. \OPT may serve each pair immediately after its arrival,
paying $\dist(p,q) = h/2^{m/2-1}$ for the connection cost. On the other hand,
\TB serves each such pair using a path that connects $p$ and $q$ in the tree
$T$. Before matching $p$ with $q$, \TB waits for a time which is at least the
length of the longest edge on this path, $h/2$ (see the analysis
in \cite{AzChKa17}). In total, the cost of \TB for the last $m$ requests alone
is at least $(m/2) \cdot (h/2)$,  while the total cost of \OPT for the whole
input is at most $h/2^m + (m/2) \cdot h/2^{m/2-1}$. This proves that the
competitive ratio of \TB is~$2^{\Omega(m)}$.
\end{proof}